\documentclass{ifacconf}

\makeatletter
\let\old@ssect\@ssect 
\makeatother

\let\theoremstyle\undefined
\usepackage{amsmath,amssymb,amsfonts,amsthm}
\usepackage{mathtools}
\usepackage{graphicx}       
\usepackage[utf8]{inputenc} 
\usepackage[T1]{fontenc}    
\usepackage[hidelinks]{hyperref}       
\usepackage{url}            
\usepackage{booktabs}       
\usepackage{nicefrac}       
\usepackage{doi}
\usepackage{nth}
\usepackage{float}
\usepackage{caption}
\usepackage{subcaption}
\usepackage{adjustbox}

\makeatletter
\def\@ssect#1#2#3#4#5#6{%
  \NR@gettitle{#6}
  \old@ssect{#1}{#2}{#3}{#4}{#5}{#6}
}
\makeatother

\usepackage[svgnames]{xcolor}
\usepackage{colortbl}%
  \newcommand{\myrowcolour}{\rowcolor[gray]{0.925}}

\usepackage{pgfplots}
\DeclareUnicodeCharacter{2212}{−}
\usepgfplotslibrary{groupplots,dateplot}
\usetikzlibrary{patterns,shapes.arrows}
\pgfplotsset{compat=newest}

\DeclareMathAlphabet{\mymathbb}{U}{BOONDOX-ds}{m}{n}

\newcommand{\R}{\mathbb{R}}

\newcommand{\bs}[1]{\boldsymbol{#1}}
\newcommand{\bb}[1]{\mymathbb{#1}}

\newcommand{\mc}[1]{\mathcal{#1}}
\newcommand{\mt}[1]{\textrm{#1}}

\newcommand{\lt}{\left}
\newcommand{\rt}{\right}
\newcommand{\beeq}{\begin{equation}}
\newcommand{\eneq}{\end{equation}}
\newcommand{\matb}{\begin{matrix}}
\newcommand{\mate}{\end{matrix}}

\DeclareMathOperator*{\argmin}{arg\,min}

\usepackage{natbib}        

\graphicspath{ {./pictures/} }

\newcounter{thms}
\newtheorem{theorem}[thms]{Theorem}

\newtheorem{lemma}[thms]{Lemma}
\newtheorem{assumption}[thms]{Assumption}

\theoremstyle{definition}

\theoremstyle{plain}
\newtheorem*{remark}{Remark}

\definecolor{mygreen}{rgb}{0.1, 0.7, 0.1}

\definecolor{bluetmp}{rgb}{0.0, 0.0, 0.0}

\begin{document}
\begin{frontmatter}

\title{Minimal regret state estimation of time-varying systems}
\date{\nth{13} December 2021}

\author[First]{Jean-S\'ebastien Brouillon}
\author[Fourth]{Florian D\"orfler}
\author[First]{Giancarlo Ferrari Trecate} 

\address[First]{Institute of Mechanical Engineering, 
	\'Ecole Polytechnique F\'ed\'erale de Lausanne (EPFL), 
	CH-1015 Lausanne, Switzerland (e-mail: jean-sebastien.brouillon@epfl.ch, giancarlo.ferraritrecate@epfl.ch)}
\address[Fourth]{Automatic Control Laboratory, 
	Swiss Federal Institute of Technology (ETH), 
	Zurich, Switzerland 
	(dorfler@control.ee.ethz.ch)}

\hypersetup{
pdftitle={Bayesian Error-in-Variables Models for the Identification of Distribution Grids},
pdfauthor={Jean-Sébastien Brouillon, Emanuele Fabbiani, Pulkit Nahata, Florian Dörfler, Giancarlo Ferrari-Trecate},
pdfkeywords={Bayesian inference, Distribution grids, Error-in-variables, Line admittance estimation, Power systems identification},
}
\thanks{This research is supported by the Swiss National Science Foundation under the NCCR Automation (grant agreement 51NF40\_180545).}

\begin{abstract}
Kalman and $\mc H_\infty$ filters, the most popular paradigms for linear state estimation, are designed for very specific specific noise and disturbance patterns, which may not appear in practice. State observers based on the minimization of regret measures are a promising alternative, as they aim to adapt to recognizable patterns in the estimation error. In this paper, we show that the regret minimization problem for finite horizon estimation can be cast into a simple convex optimization problem. For this purpose, we first rewrite linear time-varying system dynamics using a novel system level synthesis parametrization for state estimation, capable of handling both disturbance and measurement noise. We then provide a tractable formulation for the minimization of regret based on semi-definite programming. Both contributions make the minimal regret observer design easily implementable in practice. Finally, numerical experiments show that the computed observer can significantly outperform both $\mc H_2$ and $\mc H_\infty$ filters.
\end{abstract}



\end{frontmatter}

\section{Introduction}

A key contribution on the problem of estimating the state and predicting future outputs of a linear system was provided in the seminal paper by \cite{kalman_og}. This work provided the basis for $\mc H_2$ (i.e., least squares error) state estimation, which had a tremendous success in aerospace, and in many other applications later. However, for several industrial processes, the performance of $\mc H_2$ filtering can be questionable due to the poor modeling of the systems involved \citep{book_kalman_hing}. In several industrial plants, the process noise (also called disturbance) often includes different forms of uncertainty, therefore rendering the Gaussian assumption flawed. This has motivated the industry to also turn towards estimators optimized for worst case scenarios ($\mc H_\infty$), as described in the pioneering work of \cite{zames_hinf_og}. \textcolor{bluetmp}{Another approach introduced in \cite{luenberger} consists in stabilizing the error dynamics, such that errors in the initial state estimate are eventually rejected by the estimator. This approach is much simpler to implement but lacks optimal performance or robustness.} Today, the $\mc H_2$, $\mc H_\infty$, and Luenberger observers are the dominating paradigms in industrial applications \citep{book_kalman_hing}.

The $\mc H_\infty$ filter has proven to work well in practice, and it adresses the reliance of $\mc H_2$ filters on Gaussian noise and known system dynamics. However, doing so, it loses the optimality properties of $\mc H_2$ methods. The tradeoff between optimal performance and robustness can be optimized using mixed $\mc H_2$/$\mc H_\infty$ methods \citep{cont_time_mixed_h2hinf, disc_time_mixed_h2hinf}. However, the relative weights of both components are not tuned according to the disturbance patterns affecting the system, but need to be set \emph{a-priori}. In order to optimize the benefits and drawbacks of each method automatically, one can use a measure of regret compared to an observer that would know the future in advance.

Regret is the difference of performance between a practical observer and its ideal counterpart, which would anticipate all possible uncertainties a system will encounter. This ideal noncausal observer (often called clairvoyant, as it sees the future) is not implementable, but can give useful information on how to adapt to various disturbance profiles. The notion of regret has first been developed by economists \citep{regret_og}, for studying human interactions that are highly subject to changes and modelling errors. This motivation is also relevant for industrial processes, hence fostering the interest for the notion of regret measures in engineering. In the estimation setting, because the clairvoyant observer adapts to specific error patterns (e.g. constant, drifting, ...), it can compensate for errors and simplifications made during modeling. Hence, minimizing the worst case difference with this observer allows reducing the conservativeness of $\mc H_\infty$ filters, while keeping the greater robustness than $\mc H_2$.

In recent years, researchers in the field of estimation have attempted to design new observers with reduced regret. The idea of finding a versatile estimator that adapts to both Gaussian and worst-case noise through regret appears in \cite{vanli_better_regret_robust}. Nevertheless, this and most other works aim to improve regret bounds for given estimation methods. Among the works using regret measures for observer design, \cite{ouhamma_better_regret_ridge} provide an improvement to recursive ridge regression, and \cite{gharbi_mhe_regret_bound} present an Moving Horizon Estimation (MHE) method, which is suboptimal in terms of the cost function, but has a lower regret bound than the optimal estimator. The work of \cite{agrawal_improving_regret_algo} makes a step towards regret-optimal estimation, as it gives an iterative algorithm that improves the regret bound at each iteration. However, the solution presented by the authors is iterative and only optimizes regret asymptotically. So far, only \cite{babak_estim} and \cite{goel2021regret} have directly computed a state observer that minimizes regret. However, their method is only $\gamma$-suboptimal and requires specific conditions for a solution to exist. \cite{didier2022system} and \cite{martin2022safe} address this limitation in the context of regret-based controller design by formulating the minimization of regret in a convex way, and provide a simple way to introduce constraints. However, their approach cannot be applied in a straightforward way to state estimation problems, due to the presence of measurement noise.

In this paper, we provide a new convex formulation of the regret minimization problem with both process and measurement noise, which allows computing the regret-optimal observer using standard optimization tools. This formulation relies on a new parametrization adapting System Level Synthesis (SLS) \citep{sls_og_paper} for the state estimation problem. We test this observer on several systems, and show that it can improve the state estimation error compared to Kalman/$\mc H_2$ and $\mc H_\infty$ filters for various deterministic and stochastic disturbance profiles. Similar formulations for controller design show that it is possible to introduce constraints, although this aspect will not be discussed in this paper.

The paper proceeds as follows. Section \ref{sec_theory} states the problem and provides the necessary tools and assumptions. Section \ref{sec_estim} shows how to compute the $\mc H_2$, $\mc H_\infty$, clairvoyant, and regret-optimal observers in the SLS framework. Section \ref{sec_results} discusses the numerical results, and Section \ref{sec_conclu} concludes the paper.

\subsection{Preliminaries and notations}
Time indices are denoted by subscripts (e.g., $x_t$), and boldface letters denote the stacked vectors at all times, e.g. $\bs x = [x_0^\top, \dots, x_T^\top]^\top$ for a window $[0,T]$. Similarly, calligraphic letters denote linear operators applying to such stacked vectors (e.g. $\mc A$ such that $\bs y = \mc A \bs x$). The $\ell_2$-norm of a vector is denoted by $\|\cdot\|_2$, which also denotes the spectral norm of a matrix (its largest eigenvalue). The Frobenius norm of a matrix is denoted by $\|\cdot\|_F$. Blackboard bold letters denote sets (e.g. $\bb R$ for all real numbers), except $\bb E$ that denotes the expectation of a random variable.

\section{System model and problem statement}\label{sec_theory}

\subsection{State space equations}
We consider a linear time-varying system given by the equations
\begin{subequations}\label{eq_model_system}
\begin{align}
    x_{t+1} &= A_{t} x_{t} + B_{t} u_{t} + w_{t}, \\
    y_{t} &= C_{t} x_{t} + v_{t},
\end{align}
\end{subequations}
where $x_t \in \bb R^n$ is the state, $u_t \in \bb R^p$ the input, and $y_t \in \bb R^m$ the output for each time instant $t = 0, \dots, T$. The system dynamics are characterized by the time-varying parameters $A_t \in \R^{n \times n}$, $B_{t} \in \R^{n \times p}$, and $C_t \in \R^{m \times n}$, and subject to disturbance $w_t$ and measurement noise $v_t$.

The only restriction that we make on $w_t$ and $v_t$ are the following.

\begin{assumption}[Disturbance boundedness]\label{ass_bounded}
There exists non-empty full-dimensional\footnote{The ellipsoids $\bb v$ and $\bb w$ contain a $m$- and $n$-dimensional ball, respectively.} ellipsoids $\bb v = \{v : \|H_v v\|_2 \leq 1\}$ and $\bb w = \{w : \|H_w w\|_2 \leq 1\}$ such that $v_t \in \bb v$ and $w_t \in \bb w$ for all $t = 0,\dots,T$, respectively.
\end{assumption}

Assumption \ref{ass_bounded} is common in the field of $\mc H_\infty$ filtering for linear systems \cite{}. However, it can lead to poor performance as $\mc H_\infty$ state observers and filters are designed to mitigate worst case disturbances, even if they are not encountered during the system operation. This paper aims to provide a more versatile approach that adapts to the encountered disturbance patterns for better performance.

\subsection{Estimation problem}
\textcolor{bluetmp}{We are interested in computing the state predictions $\hat x_{t+1}$ of $x_{t+1}$ for $t = 0, \dots, T$, when only knowing $y_{t}$ and $u_t$ but not $v_t$ and $w_{t}$, and minimizing the total cost}
\begin{align}\label{eq_cost_def}
    \mt{cost}(\hat x_1 - x_1, \dots, \hat x_{T+1} - x_{T+1}) = \sum_{t=1}^{T+1} \ell_t(\hat x_t - x_t),
\end{align}
where $\ell_t$ is a given loss function (e.g. least squares). To do so, we design a dynamic Luenberger-type observer defined by the gains $L_{\tau|t} \in \bb R^{n \times m}$ for $\tau = 0, \dots, t$, and given by

\begin{align}\label{eq_model_def}
    \hat x_{t+1} = A_t \hat x_t + B_{t} u_{t} - \sum_{\tau=0}^t L_{\tau|t} (C_{\tau} \hat x_{\tau} - y_{\tau}).
\end{align}

The gains are chosen to stabilize the prediction error $e_t = \hat x_t - x_t$ folloing the dynamics
\begin{align}\label{eq_model_error_sys_noopt}
    e_{t+1} = A_t e_t - w_t - \sum_{\tau=0}^t L_{\tau|t} (C_{\tau} e_{\tau} - v_{\tau}),
\end{align}
while minimizing the cost \eqref{eq_cost_def}.

\begin{remark}
We highlight that the error on the initial state is embedded without loss of generality in the first disturbance, i.e. $w_0 = A_0 (\hat x_0 - x_0)$ and $e_0 = 0$.
\end{remark}

Substituting \eqref{eq_model_error_sys_noopt} in the cost \eqref{eq_cost_def} gives the final problem
\begin{align}\label{eq_model_error_sys}
    \argmin_{\bs v, \bs w, \{ L_{\tau|t}: t \in [0,T], \tau \in [0,t]\}, e} &\; cost(\bs e), \\ \vspace{-2px}
    \mt{s.t. }
    e_{t+1} &= A_t e_t - w_t -\! \sum_{\tau=0}^t L_{\tau|t} (C_{\tau} e_{\tau} \!-\! v_\tau),\nonumber
    \\
    e_0 &= 0, \nonumber
\end{align}
where
\begin{align}
    \bs v = \lt[\matb v_0 \\ v_1 \\ \vdots \\ v_T \mate\rt],
    \bs w = -\lt[\matb A_0 e_0 \\ w_1 \\ \vdots \\ w_T \mate\rt], \bs e = \lt[\matb e_1 \\ \vdots \\ e_T \\ e_{T+1} \mate\rt].
\end{align}
\begin{remark}
While in the sequel we focus on the finite horizon estimation problem \eqref{eq_model_error_sys}, we notice that it provides the building block for defining MHE methods that are applicable for $t=0,\dots,\infty$ \citep{mhe_alessandri,mhe_ferrari2002, mhe_farina}.
\end{remark}

We also define the set $\bb V$ such that each element $\bs v_i$ of $\bs v \in \bb V$ is in $\bb v$, and similarly, the set $\bb W$. From Assumption \ref{ass_bounded}, it follows that $\bb V$ and $\bb W$ are non-empty full-dimensional ellipsoids, and can therefore be written as $\bb V = \{ \bs v | \|\mc H_v \bs v\|_2 \leq 1\}$ and $\bb W = \{ \bs w | \|\mc H_w \bs w\|_2 \leq 1\}$, where $\mc H_v$ and $\mc H_w$ are square invertible block-diagonal matrices.

\subsection{System level synthesis}

In order to formulate a convex problem to compute all $L_{\tau|t}$ minimizing \eqref{eq_cost_def}, we introduce a parametrization similar to SLS \citep{sls_og_paper}, which was proposed for controller design. We will show the necessary modifications for addressing observer design in the this section.

Let $\mc Z$ be the block-downshift operator, namely a matrix with identity matrices along its first block sub-diagonal and zeros elsewhere, and let $\mc A = \mt{blkdiag}(A_{0}, \dots, A_T, 0_{n\times n})$ and $\mc C = \mt{blkdiag}(0_{m\times n}, C_{0}, \dots, C_T)$.
For $t = 0, \dots, T$, the model \eqref{eq_model_error_sys} is equivalent to
\begin{align}\label{eq_error_sys_aug}
    \bs e = \mc Z \mc A \bs e - \mc L \mc C \mc Z \bs e + \mc L \bs v + \bs w,
\end{align}
where
\begin{align}\label{eq_def_L_matrix}
    \mc L &= \lt(\matb L_{0|0} & 0 & \cdots & 0 \\  L_{0|1} & L_{1|1} & \ddots &\vdots \\ \vdots & \vdots & \ddots & 0 \\ L_{0|T} & L_{1|T} & \cdots & L_{T|T} \mate\rt)\!.
\end{align}

\begin{remark}
\textcolor{bluetmp}{
With this notation, a standard Luenberger observer is be modeled as a block-diagonal matrix with the same block $L_{t|t} = L$ repeated in all diagonal elements.
A Kalman filter is modeled as a block-diagonal $\mc L$, in which each block $L_{t|t}$ is a pre-computed Kalman gain for time $t$. The block lower-triangular formulation adopted in \cite{babak_estim} and in this paper is more general and allows one to consider noise patterns with dependencies between time instants.
}
\end{remark}

The SLS model is defined by the matrices
\begin{subequations}\label{eq_def_phi}
\begin{align}\label{eq_def_phie}
    \Phi_w &= (I - \mc Z(\mc A - \mc L \mc C))^{-1}, \\
    \Phi_v &= (I - \mc Z(\mc A - \mc L \mc C))^{-1} \mc L,
    \label{eq_def_phil}
\end{align}
\end{subequations}
where, by construction, $\Phi_w$ and $\Phi_v$ satisfy \eqref{eq_error_sys_aug}, i.e.,
\begin{align}\label{eq_trajectory_phi}
	\bs e = \Phi_w \bs w + \Phi_v \bs v.
\end{align}
The lower block-triangular matrices $\Phi_v$ and $\Phi_w$ map the noise components $v$ and $w$ to their respective shares of the estimation error $\bs e$. The maps $\Phi_v$ and $\Phi_w$ will be called prediction error maps in the rest of the paper. In contrast, when SLS is used for controller design, the closed-loop state and input maps $\Phi_x$ and $\Phi_u$ give the state and input trajectories depending solely on $\bs w$ \cite{sls_og_paper}.

\begin{theorem}\label{thm_sls}
Consider the dynamics \eqref{eq_error_sys_aug}, there exist a unique matrix $\mc L$ satisfying \eqref{eq_def_phie} and \eqref{eq_def_phil} if and only if the following achievability condition is satisfied
\begin{align}\label{eq_feas_cond}
    \Phi_w (I - \mc Z \mc A) + \Phi_v \mc C \mc Z = I.
\end{align}
Moreover $\mc L = \Phi_w^{-1} \Phi_v$.
\end{theorem}
\begin{proof}
From \eqref{eq_def_phi}, one has that $\Phi_v = \Phi_w \mc L$, and \eqref{eq_error_sys_aug} can be rewritten as
\begin{align}\label{eq_error_sys_aug_w}
    \Phi_w(\bs w + \mc L \bs v) =\;& \mc Z \mc A \Phi_w (\bs w + \mc L \bs v) - \mc L \mc C \mc Z \Phi_w (\bs w + \mc L \bs v) 
    \nonumber \\
    & + (\bs w + \mc L \bs v),
\end{align}
which is satisfied for all $\bs w \in \bb W$ and $\bs v \in \bb V$ if and only if
\begin{align}\label{eq_error_sys_aug_no_w}
    \Phi_w =\;& \mc Z \mc A \Phi_w - \mc L \mc C \mc Z \Phi_w + I.
\end{align}

In order to prove the theorem, we only need to prove that $\Phi_w$ is invertible and that \eqref{eq_feas_cond} is equivalent to \eqref{eq_error_sys_aug_no_w} with $\mc L = \Phi_w^{-1} \Phi_v$.

First, $\Phi_w \mc Z \mc A$ is lower block-triangular with a zero block-diagonal because it is the product of (i) a lower block-triangular (ii) a block downshift operator, which zeroes the block-diagonal of $\Phi_w \mc Z$, and (iii) a block diagonal matrix which does not change the sparsity pattern. Moreover, $\Phi_v \mc C \mc Z$ follows the same pattern because $\Phi_v$ and $\mc C$ are also lower block-triangular and block-diagonal, respectively. Hence, the diagonal of $\Phi_w$ must be equal to the identity on the right hand side of \eqref{eq_feas_cond}, which means that all its eigenvalues are equal to one if \eqref{eq_feas_cond} is satisfied. This proves that $\Phi_w$ is invertible.

Second, by post-multiplying and pre-multiplying \eqref{eq_error_sys_aug_no_w} by $\Phi_w^{-1}$ and $\Phi_w$, respectively, we obtain
\begin{align}\label{eq_error_sys_aug_no_w_2}
    \Phi_w = \Phi_w \mc Z \mc A - \Phi_w \mc L \mc C \mc Z + I.
\end{align}
To obtain \eqref{eq_feas_cond} from \eqref{eq_error_sys_aug_no_w_2}, one must replace $\mc L$ by $\Phi_w^{-1} \Phi_v$ and rearrange the terms, which proves the theorem.
\end{proof}

Transforming the optimal estimation problem \eqref{eq_model_error_sys} into the SLS framework gives the following convex problem.
\begin{subequations}\label{eq_model_error_sls}
\begin{align}\label{eq_model_error_sls_c}
    \argmin_{\Phi_w, \Phi_v, \bs w, \bs v} \;\;& \mt{cost}(\Phi_w \bs w + \Phi_v \bs v), \\ \vspace{-2px}
    \mt{s.t. } &
    \Phi_w (I - \mc Z \mc A) + \Phi_v \mc C \mc Z = I, \label{eq_model_error_cons}\\
    &\Phi_w, \Phi_v \mt{ lower block-triangular}. \label{eq_model_error_spars}
\end{align}
\end{subequations}

The constraint \eqref{eq_model_error_spars} is added for the prediction error maps $\Phi_v$ and $\Phi_w$ to be causal, i.e. $e_t$ only depends on $v_0, \dots, v_t$ and $w_0, \dots, w_t$ but not on future noise. Hence, the observer $\mc L = \Phi_w^{-1} \Phi_v$ is also causal. At time $t$, $\mc L$ predicts future states $x_{t+1}, \dots, x_T$ only based on past and present measurements $y_1, \dots, y_t$.

\section{Minimal regret estimation}\label{sec_estim}
We will first explain how to derive two standard causal benchmarks using the SLS, namely the optimal $\mc H_2$ and $\mc H_\infty$ observers.
\subsection{Standard benchmarks}\label{subsec_standard}
\emph{$\mc H_2$ filtering}: $\mc H_2$ observers such as the Kalman filter are designed to minimize the mean square error when disturbances and noise are Gaussian. In mathematical terms, this means that the filter gains solve the problem \eqref{eq_model_error_sys} with
\begin{align}\label{eq_kalman_opt_c}
    \mt{cost}_2(e(\bs v, \bs w, \mc L)) &= \frac{1}{T} \sum_0^T \bb E_{[v_t, w_t] \sim \mc N(0,\Sigma_t) } \lt[\|e_t\|_2^2\rt],
\end{align}
where $\Sigma_t = \mt{blkdiag}( \Sigma_{v,t}, \Sigma_{w,t})$ is the covariance of the noise $v_t$ and the disturbance $w_t$, assumed Gaussian and uncorrelated in time. Similar to \cite{sls_og_paper} and by using \eqref{eq_trajectory_phi}, one has that \eqref{eq_kalman_opt_c} can be rewritten in terms of $\Phi_w$ and $\Phi_v$ as
\begin{align}\label{eq_kalman_opt_c_phi}
    \mt{cost}_2(\Phi_v, \Phi_w) &= \frac{1}{T} \lt\|[\Phi_v, \Phi_w] \lt[ \matb \Sigma_v^{\frac{1}{2}} & 0 \\ 0 & \Sigma_w^{\frac{1}{2}} \mate \rt] \rt\|_F^2.
\end{align}
where $\Sigma_v = \mt{blkdiag}(\Sigma_{v,0},\dots, \Sigma_{v,T})\; $ and $\;\Sigma_w = \mt{blkdiag}(\Sigma_{w,0},\dots, \Sigma_{w,T})$. Theorem \ref{thm_sls} shows that one can replace the problem \eqref{eq_model_error_sys} by \eqref{eq_model_error_sls}, which finally gives.
\begin{align}\label{eq_kalman_sls}
    \{\Phi_{v,2}, \Phi_{w,2}\} = \argmin_{\Phi_{v}, \Phi_{w}} \;\;& \lt\|[\Phi_v, \Phi_w] \lt[ \matb \Sigma_v^{\frac{1}{2}} & 0 \\ 0 & \Sigma_w^{\frac{1}{2}} \mate \rt] \rt\|_F^2,
    \\ \nonumber
    \mt{s.t.}&\mt{ \eqref{eq_model_error_cons}, \eqref{eq_model_error_spars}}.
\end{align}

The optimal $\mc H_2$ observer is the Kalman filter, which is such that $\mc L$ is block-diagonal. This means that, even if $\Phi_{w,2}$ and $\Phi_{v,2}$ can be dense lower block-triangular matrices, $\mc L_2 = \Phi_{w,2}^{-1} \Phi_{v,2}$ is block diagonal. This is not included as a hard constraint in \eqref{eq_kalman_sls}. The block-diagonality of $\mc L_2$ appears naturally when $w_t$ and $v_t$ are not correlated between different time instants, because the Kalman filter is proven to be optimal \citep{kalman_og}. The problem \eqref{eq_kalman_sls} is a relaxation of general $\mc H_2$ observers, which generalizes the Kalman filter and can therefore improve performance if the noise is correlated between different instants. This means that it gives a suitable benchmark for comparison. It will also be used later to formulate regret measures.

\emph{$\mc H_\infty$ filtering}: Similar to the $\mc H_2$ observer, \cite{book_kalman_hing} describes the $\mc H_\infty$ observer as a solution to \eqref{eq_model_error_sys} with the cost
\begin{align}\label{eq_hinf_opt_c}
    \mt{cost}_\infty(e(\bs v, \bs w, \mc L)) &= \frac{1}{T} \sum_{t=0}^T \max_{\substack{v_t \in \bb v \\ w_t \in \bb w}} \|e_t\|_2^2,
    \\ \nonumber
    &= \frac{1}{T} \max_{\substack{v | \|\mc H_v \bs v\|_2 \leq 1 \\ w | \|\mc H_w \bs w\|_2 \leq 1}}\|\bs e\|_2^2.
\end{align}
By using the change of variable $\bs v_n = \mc H_v \bs v$ and $\bs w_n = \mc H_w \bs w$, one can rewrite the cost in terms of $\Phi_v$ and $\Phi_w$ as
\begin{align}\label{eq_hinf_opt_changed}
    \mt{cost}_\infty(\Phi_v, \Phi_w) \!&=\! \frac{1}{T} \!\! \max_{\substack{\|\bs v_n\|_2 \leq 1 \\ \|\bs w_n\|_2 \leq 1}}
    \lt\|[\Phi_v, \Phi_w] \!\! \lt[ \matb \mc H_v^{-1} & 0 \\ 0 & \mc H_w^{-1} \mate \rt] \!\!\! \lt[ \matb \bs v_n \\ \bs w_n \mate \rt] \! \rt\|_2^2 \!\!,
\end{align}
where $\mc H_v$ and $\mc H_w$ are invertible because of Assumption \ref{ass_bounded}. According to the definition of spectral norms, and as proven by \cite{sls_og_paper}, \eqref{eq_hinf_opt_c} is equal to
\begin{align}\label{eq_hinf_opt_c_phi}
    \mt{cost}_\infty(\Phi_v, \Phi_w) &= \frac{1}{T} \lt\|[\Phi_v, \Phi_w] \!\! \lt[ \matb \mc H_v^{-1} & 0 \\ 0 & \mc H_w^{-1} \mate \rt] \! \rt\|_2^2.
\end{align}

Hence, the $\mc H_\infty$ observer is $\mc L_\infty = \Phi_{w, \infty}^{-1} \Phi_{v, \infty}$ with
\begin{align}\label{eq_hinf_sls}
    \{\Phi_{v, \infty}, \Phi_{w, \infty}\} \in \argmin_{\Phi_{v}, \Phi_{w}} \;\;& \lt\|[\Phi_v, \Phi_w] \!\! \lt[ \matb \mc H_v^{-1} & 0 \\ 0 & \mc H_w^{-1} \mate \rt] \! \rt\|_2^2,
    \\ \nonumber
    \mt{s.t.}&\mt{ \eqref{eq_model_error_cons}, \eqref{eq_model_error_spars}}.
\end{align}
Note that the problem \eqref{eq_hinf_sls} may admit multiple solutions because the cost only affects the largest eigenvalue (i.e. the spectral norm) but not the other eigenvalues, which are free to vary. 

\subsection{Clairvoyant estimator}\label{subsec_nc}

A non-causal observer does away with the causal constraint \eqref{eq_model_error_spars}, estimating states based on both past and future measurements. This means that the optimal clairvoyant observer is given by a gain matrix $\mc L_{nc} \in \bb R^{nT \times mT}$, which is full and such that the resulting estimation error minimizes the cost \eqref{eq_cost_def}. While not implementable, $\mc L_{nc}$ can give useful information about how an observer should react to noise. In what follows, we show how to explicitly compute $\mc L_{nc}$ using the same SLS framework as in Section \ref{subsec_standard}.

\begin{lemma}\label{lem_nc}
For quadratic, positive definite loss functions $\ell_t = e_t^\top Q_t e_t$ with $0 \prec Q_t \in \bb R^{n \times n}$, the prediction error maps $\Phi_{v, nc}, \Phi_{w, nc}$ of the clairvoyant observer $\mc L_{nc} = \Phi_{w,nc}^{-1} \Phi_{v, nc}$ are given in the SLS form by 
\begin{subequations}\label{eq_nc_prob}
\begin{align}\label{eq_nc_prob_fro}
    \!\!\{\Phi_{v,nc}, \Phi_{w,nc}\} \!&= \argmin_{\Phi_{v}, \Phi_{w}} \; \lt\|\mc Q [\Phi_v, \Phi_w] \! \lt[ \matb \Sigma_v^{\frac{1}{2}} & 0 \\ 0 & \Sigma_w^{\frac{1}{2}} \mate \rt] \rt\|_F^2,
    \\ \nonumber
    &\quad\quad\quad\quad \textnormal{s.t. \eqref{eq_model_error_cons}},
    \\ \label{eq_nc_prob_spec}
    &\; \in \argmin_{\Phi_{v}, \Phi_{w}} \lt\|\mc Q [\Phi_v, \Phi_w] \!\! \lt[ \matb \mc H_v^{-1}\! & 0 \\ 0 & \!\mc H_w^{-1} \mate \rt] \! \rt\|_2^2 \!\!\!,\!
    \\ \nonumber
    &\quad\quad\quad\quad \textnormal{s.t. \eqref{eq_model_error_cons}},
\end{align}
\end{subequations}
where $\mc Q = \mt{blkdiag}(Q_0^{\frac{1}{2}}, \dots, Q_{T+1}^{\frac{1}{2}})$. The corresponding optimal cost is
\begin{align}\label{eq_nc_cost}
\!\min_{\mc L} \mt{cost}(e(\bs v, \bs w, \mc L)) = \lt\|\mc Q [\Phi_{v,nc}, \Phi_{w,nc}] \!\! \lt[ \matb \bs v  \\  \bs w \mate \rt]\! \rt\|_2^2\!
\end{align}
\end{lemma}
\begin{proof}
Using \eqref{eq_feas_cond}, one can always obtain $\{\Phi_{v,nc}, \Phi_{w,nc}\}$ from $\mc L_{nc}$ using
\begin{subequations} \label{eq_L_to_phi_nc}
\begin{align}
    \Phi_{w,nc} &= \lt( (I - \mc Z \mc A) - \mc L_{nc} \mc C \mc Z \rt)^{-1}, \\
    \Phi_{v,nc} &= \Phi_{w,nc}\mc L_{nc}.
\end{align}
\end{subequations}
The optimal non-causal estimator $\mc L_{nc}$ is a known result, which can be found in \cite{book_sayed}. It is unique and provides estimates that minimize both $\mc H_2$ and $\mc H_\infty$ costs. Using \eqref{eq_L_to_phi_nc}, this means that the corresponding pair of transfer functions $\{\Phi_{v,nc}, \Phi_{w,nc}\}$ exists, is unique, and solves both problems in \eqref{eq_nc_prob}.

The problem \eqref{eq_nc_prob_fro} is strictly convex because, by definition, $\mc Q \succ 0$, $\Sigma_v \succ 0$, and $\Sigma_w \succ 0$. Hence, its solution is unique. Since a solution solving both problems exists, and the solution to \eqref{eq_nc_prob_fro} is unique, we conclude that one only needs to solve \eqref{eq_nc_prob_fro}. Finally, \eqref{eq_nc_cost} follows from \eqref{eq_trajectory_phi}.
\end{proof}

\subsection{Regret optimal observer}
In this paper, we aim to minimize the worst-case dynamic regret, defined as 
\begin{align}\label{eq_regret_def}
    \!\!\mt{regret}(\mc L, T) \!= \!\max_{\substack{\bs v \in \bb V \\ \bs w \in \bb W}}\! \big(\mt{cost}(e(\bs v, \bs w, \mc L)) \!-\! \underbrace{\min_{\mc K}\mt{cost}(e(\bs v, \bs w, \mc K))}_{=\mt{cost}(e(\bs v, \bs w, \mc L_{nc}))} \!\! \big)\!.\!\!
\end{align}
\textcolor{bluetmp}{
With only the first term, \eqref{eq_regret_def} reduces to the $\mc H_\infty$ cost. The second term provides a system-dependent baseline, which represents the best performance achievable by the system. From Lemma \ref{lem_nc}, this baseline is provided by the clairvoyant estimator. For any noise pattern, if the clairvoyant observer gives a perfect estimate, the difference between the two costs in \eqref{eq_regret_def} is equal to the cost of $\mc L$ only. However, this difference decreases if the clairvoyant observer also generates some error. This allows the regret measure to reduce the importance of situations in which some error is unavoidable, leading to a less conservative estimator than $\mc H_\infty$. Moreover, since the clairvoyant estimator is $\mc H_2$ optimal (see \eqref{eq_nc_prob_fro}), minimizing the regret amounts to minimizing the suboptimality that a robust observer $\mc L$ induces.
}


Directly computing the minimal regret observer $\mc L_r = \argmin_{\mc L}\mt{regret}(\mc L, T)$ would be very complex due to the nested optimization problems, as well as the non-convexity of \eqref{eq_regret_def} in terms of $\mc L$. In the sequel, we compute $\mc L_r$ from $\Phi_{v}$ and $\Phi_{w}$, by adopting the approach used in \cite{martin2022safe} for designing regret-based controllers.

\begin{theorem}\label{thm_reg}
For quadratic costs and if either $\mc L_{nc}$ or $\Phi_{v,nc}$ and $\Phi_{w,nc}$ are known, the minimal regret observer is given by $\mc L_r = {\Phi_{w,r}}^{-1}\Phi_{v,r}$, where
\begin{align}\label{eq_reg_sol_def}
    \{\Phi_{v,r},\Phi_{w,r}\} = \argmin_{\Phi_v, \Phi_w} \;& \textnormal{regret}(\Phi_w^{-1} \Phi_v, T),
    \\ \nonumber
    \textnormal{s.t. }&\Phi_w^{-1} \Phi_v \textnormal{ lower block-triangular},
    \\ \label{eq_reg_sol_final}
    = \argmin_{\Phi_v, \Phi_w} \;& \min_\lambda \lambda,
\end{align}
\begin{align*}
    \textnormal{s.t. } & \textnormal{ \eqref{eq_model_error_cons}, \eqref{eq_model_error_spars}}, \lambda > 0,
    \\
    &\lt[\matb I & \mc Q^\frac{1}{2} [\Phi_v H_v^{-1}, \Phi_w H_w^{-1}]
    \\
    [\Phi_v H_v^{-1}, \Phi_w H_w^{-1}]^\top \mc Q^\frac{1}{2}
    &
    \lambda I + \mc J_{nc} \mate\rt] \succeq 0,
\end{align*}
\vspace{5px}
where $\mc J_{nc} = \lt[\matb H_v^{-\!\top} \Phi_{v,nc}^\top \\ H_w^{-\!\top} \Phi_{w,nc}^\top \mate\rt] \mc Q \lt[\matb H_v^{-\!\top} \Phi_{v,nc}^\top \\ H_w^{-\!\top} \Phi_{w,nc}^\top \mate\rt]^{\!\!\top}$.
\end{theorem}

\begin{proof}
According to Lemma \ref{lem_nc}, $\Phi_{v,nc}$ and $\Phi_{w,nc}$ give the second term of \eqref{eq_regret_def}, and can be computed from $\mc L_{nc}$. Hence \eqref{eq_regret_def} can be rewritten as
\begin{align}\label{eq_regret_def_eig}
    &\mt{regret}(\Phi_w^{-1} \Phi_v, T) = & \\ \nonumber
    &\quad \max_{\substack{\bs v \in \bb V \\ \bs w \in \bb W}} \lt[\matb v \\ w \mate\rt]^{\!\!\top} \!\! \!\Bigg(\! \lt[\matb \Phi_v^\top \\ \Phi_w^\top \mate\rt] \mc Q \lt[\matb \Phi_v^\top \\ \Phi_w^\top \mate\rt]^{\!\!\top}\!\!  -  \lt[\matb \Phi_{v,nc}^\top \\ \Phi_{w,nc}^\top \mate\rt] \mc Q \lt[\matb \Phi_{v,nc}^\top \\ \Phi_{w,nc}^\top \mate\rt]^{\!\!\top} \! \Bigg) \lt[\matb v \\ w \mate\rt] \!\!.
\end{align}
With the same change of variable as in \eqref{eq_hinf_opt_changed}, \eqref{eq_regret_def_eig} is equivalent to
\begin{align}\label{eq_regret_def_eig_2}
    &\mt{regret}(\Phi_w^{-1} \Phi_v, T) = & \\ \nonumber
    &\quad \max_{\substack{\|\bs v_n\| \leq 1 \\ \|\bs w_n\| \leq 1}} \lt[\matb v_n \\ w_n \mate\rt]^{\!\!\top} \!\! \bigg(\!
    \underbrace{\lt[\matb H_v^{-\!\top} \Phi_{v}^\top \\ H_w^{-\!\top} \Phi_{w}^\top \mate\rt] \mc Q \lt[\matb H_v^{-\!\top} \Phi_{v}^\top \\ H_w^{-\!\top} \Phi_{w}^\top \mate\rt]^{\!\!\top}
    - \mc J_{nc} }_{\mc M} \!\bigg)\!\! \lt[\matb v_n \\ w_n \mate\rt] \!\!.
\end{align}
Finding the solution of \eqref{eq_reg_sol_def} amounts to minimizing the largest eigenvalue of the matrix $\mc M$. From \cite{book_boyd}, such an eigenvalue problem is tractable and its solution is given by the semi-definite program \eqref{eq_reg_sol_final}, which proves the theorem.
\end{proof}

\vspace{-5px}

\cite{book_boyd} also show that the problem \eqref{eq_reg_sol_final} is convex and can be solved with any convex optimization methods, which makes it very convenient in practice. Note that the $\gamma$-suboptimal estimators presented in \cite{babak_estim} are cheaper to compute, but do not provide an exact solution. For time-invariant systems, the problems \eqref{eq_nc_prob_fro} and \eqref{eq_reg_sol_final} can be solved once offline, yielding the observer $\mc L_r = \Phi_{w,r}^{-1} \Phi_{v,r}$, which can be applied using \eqref{eq_model_def}. For time-varying systems, both problems must be solved every time the parameters change.

\section{Numerical experiments}\label{sec_results}
To show the benefits of minimal regret estimation, we will simulate the trajectories of the benchmark system "NN4" of order 4 from Complib (\cite{complib}) under nine different disturbance patterns. In order to assess the performance in practice, we also provide the average performance for the first three fifth order systems in Complib "AC1", "AC2", and "AC3", which model aircrafts. For convenience, we consider only LTI systems, even though our framework also encompasses LTV systems, see \eqref{eq_model_system}. All systems are discretized with $T_s = 5ms$ and the prediction horizon is $T=10$. \textcolor{bluetmp}{First, the optimal error maps $\Phi_v$ and $\Phi_w$ are computed for $\mc H_2$, $\mc H_\infty$ and minimal regret (denoted by $\mc R$) methods. Then, the performance is evaluated using \eqref{eq_trajectory_phi} and the cost $\|\bs e\|_2^2$ for each disturbance pattern.} Each element of the noises $\bs v$ and $\bs w$ follow independent uniform or Gaussian distributions in three cases, and deterministic sequences (e.g. a sine wave $w_t = [\sin(t), \dots, \sin(t)]^\top$) in the other six cases.

\begin{table}[H]
    \captionsetup{width=.45\textwidth}
    \centering
     \begin{subtable}[]{0.28\textwidth}
    \centering
    \begin{tabular}{|c||c|c|c}
     \hline
    $\bs v, \bs w$ & $\mc H_2$ & $\mc H_\infty$ & $\mc R$ \\
     \hline
    \myrowcolour
     $\mc N(0,1)$ & \textcolor{mygreen}{\textbf{1}} & 0.83\% & 11\%
     \\
     $\mc U_{[0.5,1] }$ & 7.60\% & 8.34\% & \textcolor{mygreen}{\textbf{1}}
     \\
    \myrowcolour
     $\mc U_{[0,1] }$ & 6.52\% & 7.23\% & \textcolor{mygreen}{\textbf{1}}
     \\
     1 & 7.83\% & 8.55\% & \textcolor{mygreen}{\textbf{1}}
     \\
    \myrowcolour
     sin & 7.15\% & 7.64\% & \textcolor{mygreen}{\textbf{1}}
     \\
     \!\!sawtooth\!\! & 7.72\% & 8.51\% & \textcolor{mygreen}{\textbf{1}}
     \\
    \myrowcolour
     step & 6.90\% & 8.35\% & \textcolor{mygreen}{\textbf{1}}
     \\
     stairs & 5.03\% & 4.37\% & \textcolor{mygreen}{\textbf{1}}
     \\
    \myrowcolour
     worst & 0.24\% & \textcolor{mygreen}{\textbf{1}} & \!0.24\%\!
     \\
     \hline
    \end{tabular}
    \caption{NN4}
    \label{tab_nn4}
    \end{subtable}
     \begin{subtable}[]{0.2\textwidth}
    \centering
    \begin{tabular}{c|c|c|}
    \hline
    $\mc H_2$ & $\mc H_\infty$ & $\mc R$ \\
     \hline
    \myrowcolour
    \textcolor{mygreen}{\textbf{1}} & 1.44\% & 9.70\%\!\!
     \\
    6.33\% & 4.37\% & \textcolor{mygreen}{\textbf{1}}
     \\
    \myrowcolour
    5.47\% & 3.71\% & \textcolor{mygreen}{\textbf{1}}
     \\
    6.44\% & 4.46\% & \textcolor{mygreen}{\textbf{1}}
     \\
    \myrowcolour
    5.81\% & 4.09\% & \textcolor{mygreen}{\textbf{1}}
     \\
    6.34\% & 4.35\% & \textcolor{mygreen}{\textbf{1}}
     \\
    \myrowcolour
    6.68\% & 4.27\% & \textcolor{mygreen}{\textbf{1}}
     \\
    4.02\% & 2.71\% & \textcolor{mygreen}{\textbf{1}}
     \\
    \myrowcolour
    0.26\% & \textcolor{mygreen}{\textbf{1}} & 0.41\%\!\!
     \\
     \hline
    \end{tabular}
    \caption{average of AC1-3}
    \label{tab_ac3}
    \end{subtable}
    \caption{Relative difference in average estimation costs with $\mc Q = I$, for nine difference patterns. The performance for the stochastic patterns (i.e. the first three) is averaged over 1000 realizations. The best estimator is marked by \textcolor{mygreen}{\textbf{1}} and sets the baseline for the relative percentage increase. \vspace{-6px}}
    \label{tab_results}
\end{table}

Table \ref{tab_results} shows that, as expected, $\mc H_2$ and $\mc H_\infty$ observers give the best performance for Gaussian and worst-case disturbances, respectively. However, observe that, unlike in the work of \cite{babak_estim}, the error of the $\mc R$ observer is not between that of $\mc H_2$ and $\mc H_\infty$. For all other disturbance patterns, the $\mc R$ observer achieves a significant improvement compared to the other two.

Note that the difference between the costs of $\mc H_2$ and $\mc H_\infty$ observers is relatively small compared to the difference between both costs and the one of the $\mc R$ observer. It shows that, for some systems, this new observer can provide a radically different alternative to the two state-of-the-art state estimation methods.

\section{Conclusions}\label{sec_conclu}
By using a new SLS parametrization, we provided a tractable formulation of the regret minimization problem, allowing one to compute the minimal regret observer easily. In multiple experiments, we showed that such an observer can be effective in reducing state estimation error when the disturbance and measurement noises are not Gaussian or adversarial. 

Even though linear systems are well studied, the effect of noise and disturbance distributions on estimation error is not yet fully understood. Namely, it is not trivial to find the boundaries between the sets of disturbance profiles where $\mc H_2$, $\mc H_\infty$, and minimal-regret observers are optimal. Finally, the observers proposed in the paper can provide the foundations for designing minimal-regret MHE methods for nonlinear and hybrid systems.

\bibliography{ifacconf}             
                                                   







\end{document}